\UseRawInputEncoding
\documentclass[10pt, conference, compsocconf]{IEEEtran}
\usepackage[T1]{fontenc}

\usepackage{amsmath,amssymb}
\usepackage{graphicx}
\usepackage{color}
\usepackage{amsthm}
\usepackage{bm}
\usepackage{algorithm}
\usepackage{algorithmic}
\usepackage[compress]{cite}

\newtheorem{fact}{Fact}
\newtheorem{theorem}{Theorem}

\newtheorem{definition}{Definition}

\newtheorem{remark}{Remark}

\newtheorem{example}{Example}

\newcommand{\vbu}{{\bm{v}}}
\newcommand{\ubu}{{\bm{u}}}

\newcommand{\cbu}{{\bm{c}}}

\newcommand{\hbu}{{\bm{h}}}
\newcommand{\sbu}{{\bm{s}}}

\newcommand{\ybu}{{\bm{y}}}

\ifCLASSINFOpdf
\else
\fi
\begin{document}
%
\title{A Note on ``Optimum Sets of Interference-Free Sequences With Zero Autocorrelation Zone''}


\author{\IEEEauthorblockN{Qiping Fang}
\IEEEauthorblockA{School of Cyber Engineering\\
Xidian University\\
Xi'an, China\\
qpfang@stu.xidian.edu.cn}
\and
\IEEEauthorblockN{Zilong Wang}
\IEEEauthorblockA{The State Key Laboratory of Integrated Service Networks\\
Xidian University\\
Xi'an, China\\
zlwang@xidian.edu.cn}
}


%


\maketitle

\begin{abstract}
In this paper, a simple construction of interference-free zero correlation zone (IF-ZCZ) sequence sets is proposed by well designed finite Zak transform lattice tessellation.
Each set is characterized by the period of sequences $KM$, the set size $K$ and the length of zero correlation zone  $M-1$, which is optimal with respect to the Tang-Fan-Matsufuji bound.
Secondly, the transformations that keep the properties of the optimal IF-ZCZ sequence set unchanged are given, and the equivalent relation of the optimal IF-ZCZ sequence set is defined based on these transformations.
Then, it is proved that the general construction of the optimal IF-ZCZ sequence set proposed by Popovic is equivalent to the simple construction of the optimal IF-ZCZ sequence set, which indicates that the generation of the optimal IF-ZCZ sequence set can be simplified.
Moreover, it is pointed out that the alphabet size for the special case of the simple construction of the optimal IF-ZCZ sequence set can be a factor of the period.
Finally, both the simple construction of the optimal IF-ZCZ sequence set and its special case have sparse and highly structured Zak spectra, which can greatly reduce the computational complexity of implementing matched filter banks.

\end{abstract}


%
\IEEEpeerreviewmaketitle

\section{Introduction}
	Sequence set with good correlation is of considerable interest in many applications in communication and radar system. The ideal sequence set should have perfect impulse-like auto-correlation as well as all-zero cross-correlation for all pairs of sequences. Unfortunately, the famous Welch bound \cite{Welch} implies that it is impossible to have impulse-like autocorrelation functions and all-zero cross-correlation functions simultaneously in a sequence set.

While the ideal sequence set is unattainable, an alternate compromise is to require that out-of-phase auto-correlation and cross-correlation of sequences are equal to zero in a finite zone. Such sequences are referred to as {\em zero correlation zone} (ZCZ) sequences which have found many applications in quasi-synchronous CDMA (QS-CDMA) system\cite{Fan1}, ranging system \cite{Coker}, channel estimation \cite{Islam} and spectrum-spreading \cite{Jiang}. Moreover, ZCZ sequences have been deployed as uplinked random access channel preambles in the fourth-generation cellular standard LTE \cite{4G}.

Generally, an ($N, K, Z_{cz}$) ZCZ sequence set is characterized by the period of sequences $N$, the set size $K$ and the length of zero correlation zone  $Z_{cz}$,  The upper bound given by Tang {\em et al} \cite{TangBound} implies that the parameters of a ZCZ sequence set must satisfy
\begin{equation*}
	K(Z_{cz}+1)\leq N.
\end{equation*}
This theoretical bound describes the tradeoff between the set size and the length of ZCZ for a fixed period of a sequence set. A ZCZ sequence set reaching the theoretical bound is called optimal.

 	It is remarkable that the above theoretical bound does not make any restrictions on the behavior of correlation functions outside ZCZ. We take into account ZCZ sequence sets with all-zero cross-correlations, which are referred to as {\em interference-free} ZCZ (IF-ZCZ) sequence sets in this paper. It's possibly beneficial when IF-ZCZ sequences are deployed in some interference-limited systems, such as heterogeneous cellular networks (HetNet) \cite{Rajamohan}, multi-function radar system envisaged for autonomous cars \cite{Patole}. For instance, an envisioned HetNet is composed of many low powered base stations (BS) within the coverage area of conventional BS. With the increasing density of BS in a HetNet, the signals from multiple BS in its neighborhood are superposed to the mobile terminal (MT) leading to severe interference-limited performance. MT may attept to connect to the closest low powered BS but the strongest signal may come from a conventional BS.
 On account of some given scenarios of HetNet, the delays of conventional cell signals may lay in a correlation zone whose length is hard to predict. IF-ZCZ sequence set with proper parameters is a good candidate for these scenarios.


A series of papers \cite{Brozik_B,Matsufuji,Mow} are devoted to the design of optimal IF-ZCZ sequence sets. In \cite[Equ.2]{Popovic_IF-ZAZ}, Popovi$\rm{\acute{c}}$ presented a general construction including all known constructions of optimal IF-ZCZ sequence sets. Furthermore, as the important special case of the general construction, Popovi$\rm{\acute{c}}$ presented the construction of generalized chirp-like IF-ZCZ sequences that allow particularly simple implementation of the corresponding banks of matched filters. 


In this paper, firstly, a simple construction of the optimal IF-ZCZ sequence set is given. Then, we give transformations that keep the properties of optimal IF-ZCZ sequence sets unchanged. Based on these transformations, we define the equivalent relationship of optimal IF-ZCZ sequence sets. Then, we prove that the general construction \cite{Popovic_IF-ZAZ} is equivalent to a simple construction of optimal IF-ZCZ sequence sets.
From the perspective of the simple construction of optimal IF-ZCZ sequence sets, the optimal $(KM, K, M-1)$ IF-ZCZ sequence set can be interpreted as consisting of element-by-element multiplication of $K$ discrete Fourier transform sequences of $N=KM$ with a common sequence obtained by periodically extending $K$ times a perfect sequence $\hbu_a$ of any length $M$. In addition, there is no need to consider other parameters and their limitations.
Moreover, it can be seen from the expression of the simple construction of optimal IF-ZCZ sequence sets that the alphabet size doesn't seem to be smaller than the period $N$ of the sequences. However, we indicate that the alphabet size of the general construction \cite{Popovic_IF-ZAZ} can be $KM$, which is a factor of the period $N=KM^2$ of the sequences. Finally, the simple construction of the optimal IF-ZCZ sequence set and its special case sequences both have sparse and highly structured Zak spectra, which can greatly reduce the computational complexity of implementing matched filter banks.

The rest of this paper is organized as follows. In Section \ref{1}, we introduce some basic definitions and notations of IF-ZCZ sequence set and finite Zak transform.  A simple construction of the optimal  IF-ZCZ sequence sets is given in Section \ref{2}. Transformations, equivalent relationships, and proof are given in Section \ref{3}. In Section \ref{4}, we indicate that the alphabet size of the optimal IF-ZCZ sequence sets can be reduced. The computational complexities of implementing matched filter banks are given in Section \ref{5}. In Section \ref{6}, we conclude the paper.

\section{preliminaries}\label{1}
	In this section, we introduce some basic definitions and notations of IF-ZCZ sequence set and Zak transform. Throughout the paper, $\omega_N:=e^{-2\pi \sqrt{-1}/N}$ is the $N$-th root of unity.
\subsection{ZCZ Sequence Set}
Let $\bm{\mathcal{S}}=\{\ubu_a: 0\leq a\leq K-1\}$ be a sequence set with $K$ complex sequences of period $N$. We first define the correlation of sequences. 	
\begin{definition}
	Let $\ubu_a=(u_a(0),u_a(1),\cdots,u_a(N-1))$ and $\ubu_b=(u_b(0),u_b(1),\cdots,u_b(N-1))$  be two complex sequences of length $N$ in $\bm{\mathcal{S}}$. The {\em cross-correlation} between $\ubu_a$ and $\ubu_b$ at shift $n$ is defined by
	\begin{equation}\label{crosscorrelation}
		\theta_{a,b}(n)=\sum_{m=0}^{N-1}
		u_a(m)u_b^*(m-n), \qquad 0\leq n \leq N-1
	\end{equation}
	where $m-n$ is taken modulo $N$ and the symbol $*$ denotes the complex conjugation. When $a=b$, the {\em auto-correlation} of $\ubu_a$ at  shift $n$ is denoted by  $\theta_{a}(n)=\theta_{a,a}(n)$.
\end{definition}

\begin{definition}
	A sequence $ \ubu_a=(u_a(0),u_a(1),\cdots,u_a(N-1))$ is referred to be {\em perfect} if
	$\theta_{a}(n)=0$ for $0<n\leq N-1$.
\end{definition}

\begin{definition}
	The length $Z_{cz}$ of {\em zero correlation zone} (ZCZ)
	of the set $\bm{\mathcal{S}}$ is defined by
	\begin{align*}
		Z_{cz}=max\{T|\theta_{a,b}(n)=0\ (0\leq |n|\leq T \ \\
		\text{for} \  a\neq b \  \text{and} \ 0<|n|\leq T \ \text{for} \  a=b ) \}.
	\end{align*}
	$\bm{\mathcal{S}}$ is referred to as an $(N,K,Z_{cz})$-ZCZ sequence set. It is clear that $|\theta_{\ubu,\vbu}(-n)|=|\theta_{\vbu,\ubu}(n)|$. To determine the value of $Z_{cz}$, it suffices to compute the correlation function of sequences $\ubu_a$ and $\ubu_b$ in set $\bm{\mathcal{S}}$ with $0\leq n \leq N-1$, i.e.,
	\begin{align*}
		Z_{cz}=max\{T|\theta_{a,b}(n)=0\ (0\leq n\leq T \ \\
		\text{for} \  a\neq b \  \text{and} \ 0<n\leq T \ \text{for} \  a=b ) \}.
	\end{align*}
\end{definition}

\begin{definition}
	An $(N,K,Z_{cz})$-ZCZ sequence set $\bm{\mathcal{S}}$ is referred to as $(N,K,Z_{cz})$ {\em interference-free} (IF) ZCZ sequence set if the cross-correlation of any two distinct sequences in $\bm{\mathcal{S}}$ is always zero. 
\end{definition}

In addition to the traditional method to study the sequences in time and Fourier space, another approach to study the sequences in Zak space was proposed in \cite{Brozik_A,Brozik_B,Brozik_Bat,Brozik_Characterization}.

\subsection{Finite Zak Transform}
Suppose for the remainder of this paper that $N=LM$, where $L$ and $M$ are
positive integers, and set $n=rM+k$, $m=iL+j$ for $0\leq k, i\leq M-1$, $0\leq r, j\leq L-1$.
\begin{definition}\label{def3}
	The $L\times M$ {\em finite Zak transform} (FZT) \cite{Zak} of sequence $\ubu=(u(0),u(1),\cdots,u(N-1))$ is defined by
	\begin{equation}\label{FZT}
		U(j,k)=\sum_{r=0}^{L-1}
		u(rM+k)\omega_L^{rj}.
	\end{equation}
\end{definition}
If we set $U(j,k)$ as the entry of an $L\times M$ matrix ${\bm U}$ at row $j$ and column $k$, it is much simple to understand the FZT by the product of the matrices. The sequence $\ubu$ can be re-expressed by an $L\times M$ matrix $\bm{A}^{\ubu}$, where the entry $A^{\ubu}(r,k)=u(rM+k)$, i.e,
\begin{equation*}
	{\bm{A}^{\ubu}}=
	\left(
	\begin{array}{cccc}
		u(0)     & u(1)     & \cdots & u(M-1)\\
		u(M)    & u(M+1) & \cdots & u(2M-1)\\
		\vdots & \vdots & \ddots & \vdots\\
		u(N-M) & u(N-M+1) & \cdots &u(N-1)\\
	\end{array}
	\right).
\end{equation*}
Let $\bm{F}$ be the DFT matrix of order $L$. It is straightforward that
\[
\bm{U}=\bm{F}\cdot{\bm{A}^{\ubu}}.
\]
If $L=N$ and $M=1$, the FZT  is identical to the DFT. If $L=1$ and $M=N$, the FZT of $\ubu$ is identical to the original sequence $\ubu$. Thus, the FZT is a primary time-frequency representation which can concurrently encodes both the time and the frequency information about a sequence.

Similar to the inverse DFT, we can define the inverse FZT, where
\begin{equation}\label{IFZT}
	u(rM+k)=L^{-1}\sum_{j=0}^{L-1}
	U(j,k)\omega_L^{-rj}.
\end{equation}
The sequence $\ubu$  can be recovered from the inverse FZT.

\subsection{FZT and Correlation of Sequences}	
The relationship between the Zak spectra and the Fourier spectra  of the sequence $\ubu$ can be given by the following formula.
\begin{equation}\label{DFT_of_FZT}
	\widehat{u}(iL+j)=
	\sum_{k=0}^{M-1}U(j,k)\omega_{N}^{jk}\omega_{M}^{ik}.
\end{equation}

For sequnces $\ubu_a$ and $\ubu_b$ of length $N$, their cross-correlation $\theta_{a, b}(n)$ at shift $n$ can be calculated by Definition \ref{crosscorrelation}. We take $\{\theta_{a, b}(0), \theta_{a, b}(1), \cdots, \theta_{a, b}(N-1)\}$ as a sequence of length $N$, denoted by ${\bm \theta_{a,b}}$. The relationship between the Zak spectra of $\ubu_a$ and $\ubu_b$ and the Zak spectra of their correlation sequence ${\bm \theta_{a,b}}$ is given below.

\begin{fact}\label{fact2}
	(\cite{Brozik_A,Brozik_B,Brozik_Bat,Brozik_Characterization}) Let ${\bm U}_a$, ${\bm U}_b$ and ${\bm \Theta_{a,b}}$  be the FZTs of sequences $\ubu_a$, $\ubu_b$, ${\bm \theta_{a,b}}$, respectively. Then we have
	\begin{equation}\label{Brozik's_FZT_of_crosscorrelation}
		\Theta_{a,b}(j,k)=\sum_{l=0}^{M-1}U(j,l)V^*(j,l-k).
	\end{equation}
\end{fact}
Note that $l-k$ in formula (\ref{Brozik's_FZT_of_crosscorrelation}) may be less than 0, which is not well defined. On the other hand, by extending the definition of the FZT, it is shown that FZT is quasi-periodic in the time variable \cite{Brozik_A}, i.e,
\begin{equation}
	U(j,k+M)=\omega_L^{-j}U(j,k).
\end{equation}
It is much clear to re-express the formula (\ref{Brozik's_FZT_of_crosscorrelation}) as following:
\begin{equation}\label{crosscorrelation in ZS}
	\begin{split}
	\Theta_{a,b}(j,k)=&\omega_L^{-j}\sum_{l=0}^{k-1}
	U(j,l)V^*(j,l-k+M)\\
	&+\sum_{l=k}^{M-1}U(j,l)V^*(j,l-k).
	\end{split}
\end{equation}

\section{A simple construction of optimal IF-ZCZ sequence sets}\label{2}
\textbf{Construction}: For positive integers $K$ and $M$, $\bm{\mathcal{S}_1}$ is a set containing $K$ sequences with period $N=KM$. Let $\ubu_a$ be the $a$th sequence of $\bm{\mathcal{S}_1}$, whose $n$th ($n=rM+k$) element is defined as
\begin{equation}\label{simple}
	u_a(n)=h_a(n\ {\rm mod}\ M)\omega_N^{an}, 
\end{equation}
where $0\leq n \leq N-1$ and $\hbu_a=(h_a(0),\dots,h_a(M-1))$ is an arbitrary perfect sequence of length $M$ for $0\leq a\leq K-1$. 

\begin{theorem}\label{th2}
	Sequence $\ubu_a$ in the set $\bm{\mathcal{S}}_1$ has sparse and highly structured $K\times M$ Zak spectra:
	\begin{equation}\label{FZT_of_construction1}
		U_a(j,k)=
		\begin{cases}
			Kh_a(k)\omega_N^{ak},&\text{if $(j+a)\mod K=0$,}\\
			0,&\text{else}
		\end{cases}
	\end{equation}
	for $0 \leq j\leq K-1, 0 \leq k \leq M-1$.
\end{theorem}
Note that the entries of Zak spectra matrix ${\bm S}_a$ are all zeroes except the $(K-a)$th row.
\begin{proof}
Let $L=K$, we first compute the $K\times M$ FZT of the sequence $\sbu_a$ in construction 1. By the definition of $L\times M$ FZT,  we have
\begin{equation}\label{FZT-1}
	\begin{split}
			U_a(j,k)
		&=\sum_{r=0}^{K-1}
		u_a(k+rM)\omega_K^{rj}\\
		&=\sum_{r=0}^{K-1}
		h_a(k)\omega_N^{a(k+rM)}\omega_K^{rj}\\
		&=h_a(k)\omega_N^{ak}
		\sum_{r=0}^{K-1}\omega_K^{r(a+j)}\\
		&=Kh_a(k)\omega_N^{ak}\delta(j+a-K),
	\end{split}
\end{equation}
where $\delta(\cdot)$ is the delta function such that $\delta(0)=1$ and $\delta(j)=0$ for $j\neq 0$. 
\end{proof}

%
%
%
%

\begin{theorem}\label{th1}
	$\bm{\mathcal{S}}_1$ is an optimal $(N,K,M-1)$ IF-ZCZ  sequence set.
\end{theorem}
\begin{proof}
The FZT of the cross-correlation of sequences $\sbu_a$ and $\sbu_b$ in construction 1 can be calculated by applying the formula (\ref{crosscorrelation in ZS}).

If $a\neq b$, the none zero elements must occur in different rows of their respective Zak spectra matrix ${\bm U}_a$ and ${\bm U}_b$, we immediately have
\[
\Theta_{\ubu_a,\ubu_b}(j,k)=0
\]
for $0 \leq j\leq K-1, 0 \leq k \leq M-1$. 	
By using the IFZT in (\ref{IFZT}), we obtain the cross-correlation of $\sbu_a$ and $\sbu_b$:
\begin{equation*}
	\theta_{\ubu_a,\ubu_b}=\bm{0}.
\end{equation*}
Thus the sequence set $\bm{\mathcal{S}}_1$ must be interference free.

If $a=b$ and $j\neq K-a$, we have Zak spectrum $U_a(j,k)=0$ for  $\forall k$. Then $\Theta_{\ubu_a}(j,k)=0$ follows.

If $a=b$ and $j=K-a$, from the Zak spectra in (\ref{FZT-1}), we have
$$U_a(j,k)=Kh_a(k)\omega_N^{k}.$$
Then
\begin{equation*}
	\begin{split}
		\Theta_{\ubu_a}(j,k)
		&=\omega_L^{-j}\sum_{l=0}^{k-1}
		U_a(j,l)U_a^*(j,l-k+M)\\
		&+\sum_{l=k}^{M-1}U_a(j,l)U_a^*(j,l-k)\\
		&=K^2\omega_K^{a-K}\sum_{l=0}^{k-1}
		h_a(l)\omega_N^{al}
		h_a^*(l-k+M)\omega_N^{-a(l-k+M)}\\
		&+K^2\sum_{l=k}^{M-1}h_a(l)\omega_N^{al}
		h_a^*(l-k)\omega_N^{-a(l-k)}\\
		&=K^2\omega_N^{ak}
		\sum_{l=0}^{M-1}h_a(l)
		h_a^*(l-k)\\
		&=NK\omega_N^{ak}\delta(k),
	\end{split}
\end{equation*}
which means the entries of $\Theta_{\ubu_a}$ are all zeroes except $\Theta_{\ubu_a}(K-a, 0)$.

By using the IFZT in (\ref{IFZT}),  the periodic auto-correlation of $\sbu_a$ is given below.
\begin{equation}\label{cor1}
	\begin{split}
		\theta_{\ubu_a}(k+rM)
		&=K^{-1}\sum_{j=0}^{K-1}
		\Theta_{\ubu_a}(j,k)\omega_K^{-rj}\\
		&=N\delta(k)\sum_{j=0}^{K-1}\delta(j+a-K)\omega_K^{-rj}\\
		&=N\omega_K^{ra}\delta(k).
	\end{split}
\end{equation}
In other word,
\begin{equation*}
	\theta_{\ubu_a}(n)=
	\begin{cases}
		N\omega_K^{\frac{na}{M}},&\text{if $M\mid n$,}\\
		0,&\text{else}.
	\end{cases}
\end{equation*}
So the length of ZCZ in $\bm{\mathcal{S}}_1$  is $M-1$, which satisfy the Tang-Fan-Matsufuji bound.
\end{proof}

\begin{example}\label{ex1}
	
	Let $K=M=4$, $N=16$, $\hbu_a$=(1,1,1,-1) for $0\leq a \leq 3$. There are four sequences in the set $\bm{\mathcal{S}}_1$:
	\begin{equation*}
		\begin{split}
			&\ubu_0=(1,1,1,-1,1,1,1,-1,1,1,1,-1,1,1,1,-1),\\
			&\ubu_1=(\omega_{16}^{0},\omega_{16}^{1},\omega_{16}^{2},\omega_{16}^{11},
			\omega_{16}^{4},\omega_{16}^{5},\omega_{16}^{6},\omega_{16}^{15},
			\omega_{16}^{8},\omega_{16}^{9},\\
			&\omega_{16}^{10},\omega_{16}^{15},
			\omega_{16}^{12},\omega_{16}^{13},\omega_{16}^{14},\omega_{16}^{7}),\\
			&\ubu_2=(\omega_{8}^{0},\omega_{8}^{1},\omega_{8}^{2},\omega_{16}^{7},
			\omega_{8}^{4},\omega_{8}^{5},\omega_{8}^{6},\omega_{8}^{11},
			\omega_{8}^{0},\omega_{8}^{1},\omega_{8}^2,\omega_{8}^7,\\
			&\omega_{8}^{4},\omega_{8}^{5},\omega_{8}^{6},\omega_{8}^{11}),\\
			&\ubu_3=(\omega_{16}^{0},\omega_{16}^{3},\omega_{16}^{6},\omega_{16}^{1},
			\omega_{16}^{12},\omega_{16}^{15},\omega_{16}^{2},\omega_{16}^{13},
			\omega_{16}^{8},\omega_{16}^{11},\\
			&\omega_{16}^{14},\omega_{16}^{9},
			\omega_{16}^{4},\omega_{16}^{7},\omega_{16}^{10},\omega_{16}^{5}).
		\end{split}
	\end{equation*}
	From Theorem \ref{th2}, Zak spectra of the above sequences can be respectively given below.
		\[
	\begin{matrix}
		{\bm U_0}=4\begin{bmatrix}
			1  & 1 & 1 & -1\\
			0  & 0    &   0    & 0\\
			0  & 0    &   0    & 0\\
			0  & 0    &   0    & 0\\
		\end{bmatrix}, & {\bm U_1} =4\begin{bmatrix}
			0  & 0    &   0    & 0\\
			0  & 0    &   0    & 0\\
			0  & 0    &   0    & 0\\
			1  & \omega_{16}^{1} & \omega_{16}^{2} & \omega_{16}^{11}\\
		\end{bmatrix}.
	\end{matrix}
	\]
	\[
	\begin{matrix}
		{\bm U_2}=4\begin{bmatrix}
			0  & 0    &   0    & 0\\
			0  & 0    &   0    & 0\\
			1  & \omega_{8}^1    &   \omega_{8}^2    & \omega_{8}^7\\
			0  & 0    &   0    & 0\\
		\end{bmatrix}, & {\bm U_3} =4\begin{bmatrix}
			0  & 0    &   0    & 0\\
			1  & \omega_{16}^3    &   \omega_{16}^{6}    & \omega_{16}^{1}\\
			0  & 0    &   0    & 0\\
			0  & 0    &   0    & 0\\
		\end{bmatrix}.
	\end{matrix}
	\]
	Without loss of generality, we  calculate the FZT of
	the cross-correlation of $\sbu_0$ and $\sbu_1$ and the FZT of
	auto-correlation of $\sbu_0$ by formula (\ref{crosscorrelation in ZS}):
		\[
	\begin{matrix}
		{\bm \Theta_{\ubu_0,\ubu_1}}=\begin{bmatrix}
			0  & 0    &   0    & 0\\
			0  & 0    &   0    & 0\\
			0  & 0    &   0    & 0\\
			0  & 0    &   0    & 0\\
		\end{bmatrix}, & {\bm \Theta_{\ubu_0}} =16\begin{bmatrix}
			1  & 0    &   0    & 0\\
			0  & 0    &   0    & 0\\
			0  & 0    &   0    & 0\\
			0  & 0    &   0    & 0\\
		\end{bmatrix}.
	\end{matrix}
	\]
	
	By the inverse FZT, we can obtain the cross-correlation ${\bm \theta_{\ubu_0,\ubu_1}}$ and the auto-correlation ${\bm \theta_{\ubu_0}}$ from ${\bm \Theta_{\ubu_0,\ubu_1}}$ and ${\bm \Theta_{\ubu_0}}$, respectively:
	\[
	{\bm \theta_{\ubu_0,\ubu_1}}={\bm 0},
	\]
	\[
	{\bm \theta_{\ubu_0}}=16(1,0,0,0,1,0,0,0,1,0,0,0,1,0,0,0).
	\]
	Then it is clear that the sequence set $\{\sbu_a|0\leq a \leq 3 \}$ is an optimal $(16,4,3)$ IF-ZCZ sequence set.
\end{example}

\section{Transformation of IF-ZCZ sequence sets}\label{3}
We consider simple operations which leave IF-ZCZ properties and sequence length invariant.
An optimal $(N,K,M-1)$ IF-ZCZ sequence set $\bm{\mathcal{S}}=\{\ubu_a: 0\leq a\leq K-1\}$ remains   an optimal IF-ZCZ sequence set under the following operations:
\begin{itemize}
	\item $\{\ubu_{\pi(a)}: 0\leq a\leq K-1\}$ where $\pi$ is a permutation of the set $\{0, 1, \cdots, K-1\}$.
	\item $\{u_a(n+\tau_a): 0\leq a\leq K-1,0\leq n \leq N-1\}$ where $\tau_a$ is shift with $0\leq \tau_a \leq N-1$.
	\item $\{\ubu_a e^{-2\pi \sqrt{-1}g(a)}: 0\leq a\leq K-1\}$ where $g(a)$ are any rational numbers.
\end{itemize}
The actions on an optimal IF-ZCZ sequence set generate a set of optimal IF-ZCZ sequence sets, which we call the transformations of the original optimal IF-ZCZ sequence set.

\begin{definition}
	If the optimal IF-ZCZ sequence set $\bm{\mathcal{S'}}$ is obtained by the transformations of the optimal IF-ZCZ sequence set $\bm{\mathcal{S}}$, then $\bm{\mathcal{S'}}$ is equivalent to $\bm{\mathcal{S}}$.
\end{definition}

Here we revisit the general construction of optimal IF-ZCZ sequence sets proposed in \cite{Popovic_IF-ZAZ}. $\bm{\mathcal{S}_0}$ is a set containing  $K$ interference-free sequence of length $N=KM$. The $n$th element of sequence $\ubu_a$ is defined as 
\begin{equation}\label{general}
	u_a(n)=h_a(n\ {\rm mod}\ M)\omega_{N}^{(ea+t)n}, 
\end{equation}
where $0\leq n \leq N-1, 0\leq a \leq K-1$, and $\hbu_a=(h_a(0),\dots,h_a(M-1))$ is an arbitrary perfect sequence of length $M$, and $e$ and $t$ are integers such that $gcd(e, K)=1$.

\begin{theorem}
	Each optimal IF-ZCZ sequence set $\bm{\mathcal{S}_0}$ of the general construction \cite{Popovic_IF-ZAZ} is equivalent to the optimal IF-ZCZ sequence set $\bm{\mathcal{S}_1}$.
\end{theorem}

\begin{proof}
	The general construction can be re-expressed in the following manner.
	For positive integers $K$ and $M$, each sequence $\ubu_a$ in \cite{Popovic_IF-ZAZ} is based on the following perfect sequence and functions:
	\begin{itemize}
		\item[(1)] $\hbu_a=(h_a(0),\dots,h_a(M-1))$ is an arbitrary perfect sequence of length $M$ for $0\leq a\leq K-1$.
		\item[(2)] $\pi$ is a permutation of the set $\{0, 1, \cdots, K-1\}$.
		\item[(3)] $f$ is a function from the set $\{0,1,\dots,K-1\}$ to the set $\{0,1,\dots,N-1\}$ such that $f(a)\equiv\pi(a)$ (mod $K$).
	\end{itemize}
	The $n$th element of sequence $\ubu_a$ is defined as	
	\begin{equation}\label{construction2}
		u_a(n=rM+k)=h_a(k)\omega_N^{f(a)k+rM\pi(a)},
	\end{equation}
	for $0\leq k\leq M-1$ and $0\leq a,r \leq K-1$.

	On one hand, if we set $f(a)=ea+t$ for gcd$(e, K)=1$ and $\pi(a)\equiv f(a)$ (mod $K$), (\ref{construction2}) can be re-expressed by
	$$u_a(n)=h_a(n\ {\rm mod}\ M)\omega_{N}^{(ea+t)n},$$
	which is exactly the expression of sequences in \cite{Popovic_IF-ZAZ}.

	On the other hand, it is clear that the set $\{\ubu_{\pi(a)}\}$ is equivalent to the set $\{\ubu_a\}$ for $0\leq a\leq K-1$, so we can always set $\pi(a)=a$ and $f(a)=t_aK+a$ in (\ref{construction2}) by the following expression:
	\begin{equation}\label{Mid}
		u_a(n=rM+k)=h_a(k)\omega_N^{(t_aK+a)k+rMa}=(h_a(k)\omega_M^{t_ak})\omega_N^{an}.
	\end{equation}

	By definition, the $n$th element of the auto-correlation of $\{h_a(k)\omega_M^{t_ak}\}$ is defined as
	\[
	\begin{split}
		\theta_a(n)
		&=\sum_{k=0}^{M-1}h_a(k)\omega_M^{t_ak}h_a^*(k-n)\omega_M^{-t_a(k-n)}\\
		&=\omega_M^{t_an}\sum_{k=0}^{M-1}h_a(k)h_a^*(k-n)
	\end{split}
	\]
	where $0\leq n,k \leq M-1$. Since $\{h_a(k)\}$ is a perfect sequence, we have $\sum_{k=0}^{M-1}h_a(k)h_a^*(k-n)=0$ for $0\leq k \leq M-1$. It is obvious that $\theta_a(n)=0$ for $0< n \leq M-1$. This means that $\{h_a(k)\omega_M^{t_ak}\}$ is perfect. Therefore, sequences (\ref{Mid})  can be simplified by
	\begin{equation*}
		u_a(n=rM+k)=h_a(k)\omega_N^{an}.
	\end{equation*}
	which is exactly the expression of optimal IF-ZCZ sequences in the simple construction.
\end{proof}

Compared with the general construction, there is no need to consider the parameters $e, t$ and their constraints. Since the parameters $e$ and $t$ in the general construction do not produce essentially different sequence sets, but only replace the index of the sequence set. 
The construction (\ref{simple}) can be interpreted as consisting of element-by-element multiplication of $K$ discrete Fourier transform sequences of $N=KM$ with a common sequence obtained by periodically extending $K$ times a perfect sequence $\hbu_a$ of any length $M$. 

Generally speaking, perfect sequences with constant magnitude can be classified into four classes:
\begin{itemize}
	\item generalized Frank sequences duo to Kumar, Scholtz, and Welch \cite[Theoerm 3]{Kumer},
	\item generalized chirp-like polyphase sequences due to Popovic \cite{Popovic chirp like},
	\item Milewski sequences \cite{Milewski},
	\item the general construction of generalized bent function due to Chung and Kumar \cite{Chung}.
\end{itemize}
Further, the above four classes are special cases of a unified construction of perfect sequences with constant magnitude proposed by Mow \cite{Mow}. Besides, perfect sequences with non-constant magnitude can be obtained from the discrete Fourier transform of the constant magnitude sequences. Especially, in the fourth-generation cellular standard LTE \cite{4G}, the random access preambles are generated from $(\ref{simple})$ by Zadoff-Chu sequences.


%

\section{Smaller Alphabet Size}\label{4}
As an important special case of the general construction \cite{Popovic_IF-ZAZ}, Popovic presented the construction of generalized chirp-like IF-ZCZ sequences that allow particularly simple implementation of the corresponding banks of matched filters. From the expression of the simple construction of optimal IF-ZCZ sequence sets, it seems that the alphabet size of the sequence sets is not less than period $N$. Actually, by selecting appropriate perfect sequences with period $M^2$, 
we indicate that the alphabet size of the simple construction of optimal IF-ZCZ sequence sets can be $KM$, which is a factor of period $N=KM^2$. 

For positive integers $K$ and $M$, $\bm{\mathcal{S}_2}$ is a set containing $K$ sequences with period $N=KM^2$. Let $\ubu_a$ be the $a$th sequence of $\bm{\mathcal{S}_2}$, whose $n$th ($n=rM+k$) element is defined as
\begin{equation}
	u_a(n=rM+k)=\omega_{KM}^{r(K\sigma_a(k)+a)},
\end{equation}
where $0 \leq r \leq KM-1$, $0 \leq k \leq M-1$, and $\sigma_a$ is a permutation of the set $\{0, 1, \cdots, M-1\}$ for $0\leq a \leq K-1$. 

\begin{theorem}\label{th4}
	Sequence $\sbu_a$ in the set $\bm{\mathcal{S}}_2$ has sparse and highly structured $KM\times M$ Zak spectra:
	\begin{equation}\label{FZT_of_construction1}
		U_a(j,k)=
		\begin{cases}
			KM,&\text{if $j=KM-K\sigma_a(k)-a$,}\\
			0,&\text{else}
		\end{cases}
	\end{equation}
	where $0 \leq j\leq KM-1, 0 \leq k \leq M-1$.
\end{theorem}
\begin{proof}
By the definition of $KM\times M$ FZT,  we have
\begin{equation}\label{FZT-1}
	\begin{split}
		U_a(j,k)
		&=\sum_{r=0}^{KM-1}
		u_a(rM+k)\omega_{KM}^{rj}\\
		&=\sum_{r=0}^{KM-1}
		\omega_{KM}^{r(K\sigma_a(k)+a+j)}\\
		&=KM\delta(j+K\sigma_a(k)+a-KM).
	\end{split}
\end{equation}
for $0 \leq j\leq KM-1, 0 \leq k \leq M-1$. 
\end{proof}

Note that for sequence $\sbu_a$,  the position $(j, k)$ of the none-zero elements in matrix $\bm{U}_a$ must satisfy $j+K\sigma_a(k)+a=KM$. For each column $k$,  there is a unique $j$ such that $$j=KM-K\sigma_a(k)-a,$$ so there is only one none-zero element in each column of matrix $\bm{U}_a$. Moreover, for each row $j$, there is a unique none-zero element in this row if and only if
$$j+a\equiv 0 \  (\mbox{mod} \ K).$$ Otherwise, the elements in row $j$ must be all zeroes.

\begin{theorem}
	$\bm{\mathcal{S}_2}$ is an optimal $(N,K,M^2-1)$ IF-ZCZ sequence set with alphabet size $KM$ which is a special case of the simple construction (\ref{simple}).
\end{theorem}

\begin{proof}
	We choose a generalized modulatable Frank sequence with period $M^2$ as $\hbu_a$, whose $n$th element is defined as
	\begin{equation}\label{Frank sequence}
		h_a(n)=c_a(n\ mod \ M)\omega_{M^2}^{-a(n\ mod \ M)}\omega_M^{\lfloor n/M \rfloor\sigma_a(n \ mod \ M)},
	\end{equation}
	where $\cbu_a=(c_a(0),c_a(1),\cdots,c_a(M-1))$ is a complex sequence with unit magnitude elements of period $M$ and $\sigma$ is a permutation of the set $\{0, 1, \cdots, M-1\}$ for $0\leq a \leq K-1$. By setting $s=1,r_0=n_0=n_1=0$ in \cite[Theorem 1]{Mow}, we can obtain a generalized Frank sequence $\{\omega_M^{\lfloor n/M \rfloor\pi_a(n \ mod \ M)}\}$ 
	which is periodically modulated with an arbitrary sequence $\{c_a(n\ mod \ M)\omega_N^{-a(n\ mod \ M)}\}$ of period $M$.
	
	By inserting $(\ref{Frank sequence})$ in $(\ref{simple})$, we obtain
	\begin{equation}\label{construcion 1}
		u_a(n=rM+k)=c_a(k)\omega_{KM}^{r(K\sigma_a(k)+a)},
	\end{equation}
	for $0 \leq r \leq KM-1$ and $0 \leq k \leq M-1$. 
	
	Let $\cbu_a=1$. We have
	\[
	u_a(n=rM+k)=\omega_{KM}^{r(K\sigma_a(k)+a)},
	\]
	where the alphabet size is $KM$.
\end{proof}

	\begin{example}\label{ex2}
	Let $K=2,M=4,N=32$,  $\sigma_0(k)=k$ and $\sigma_1(0,1,2,3)=(1,3,2,0)$. There are two sequences in the set $\bm{\mathcal{S}}$:
	\[
	\begin{split}
		\ubu_0=
		(&\omega_8^{0},\omega_8^{0},\omega_8^{0},\omega_8^{0},
		\omega_8^{1},\omega_8^{3},\omega_8^{5},\omega_8^{7},\\
		&\omega_8^{2},\omega_8^{6},\omega_8^{2},\omega_8^{6},
		\omega_8^{3},\omega_8^{1},\omega_8^{7},\omega_8^{5},\\
		&\omega_8^{4},\omega_8^{4},\omega_8^{4},\omega_8^{4},
		\omega_8^{5},\omega_8^{7},\omega_8^{1},\omega_8^{3},\\
		&\omega_8^{6},\omega_8^{2},\omega_8^{6},\omega_8^{2},
		\omega_8^{7},\omega_8^{5},\omega_8^{3},\omega_8^{1}),\\
	\end{split}
	\]
	\[
	\begin{split}
		\ubu_1=
		(&\omega_4^{0},\omega_4^{0},\omega_4^{0},\omega_4^{0},
		\omega_4^{2},\omega_4^{0},\omega_4^{3},\omega_4^{1},\\
		&\omega_4^{0},\omega_4^{0},\omega_4^{2},\omega_4^{2},
		\omega_4^{2},\omega_4^{0},\omega_4^{1},\omega_4^{3},\\
		&\omega_4^{0},\omega_4^{0},\omega_4^{0},\omega_4^{0},
		\omega_4^{2},\omega_4^{0},\omega_4^{3},\omega_4^{1},\\
		&\omega_4^{0},\omega_4^{0},\omega_4^{2},\omega_4^{2},
		\omega_4^{2},\omega_4^{0},\omega_4^{1},\omega_4^{3}).
	\end{split}
	\]
	
	From definition (\ref{def3}) , their Zak spectra can be respectively given below.
	\[
	\begin{matrix}
		{\bm U_0}=8\begin{bmatrix}
			0  & 0    &   0     &  0  &   0  &  0  &  0  &  1 \\
			0  & 0    &   0     &  0  &   0  &  1  &  0  &  0  \\
			0  & 0    &   0    &   1   &  0  &  0  &  0  &  0 \\
			0  & 1    &   0    &   0  &   0  &  0  &  0  &  0 \\
		\end{bmatrix}^T,
	\end{matrix}
	\]
	\[
	\begin{matrix}
		 {\bm U_1} =8\begin{bmatrix}
			0  & 0    &   0     &  0  &   1  &  0  &  0  &  0 \\
			1  & 0    &   0     &  0  &   0  &  0  &  0  &  0  \\
			0  & 0    &   1    &   0   &  0  &  0  &  0  &  0 \\
			0  & 0    &   0    &   0  &   0  &  0  &  1  &  0 \\
		\end{bmatrix}^T.
	\end{matrix}
	\]
	The FZTs of cross-correlation ${\bm \theta_{1,2}}$ and auto-correlation ${\bm \theta_{1}}$ and ${\bm \theta_{2}}$  are given as following.
	\[
	\begin{matrix}
		{\bm \Theta_{0,1}}=\begin{bmatrix}
			0  & 0    &   0     &  0   &  0  &  0  &  0  &   0  \\
			0  & 0    &   0     &  0  &  0  &  0  &  0  &  0  \\
			0  & 0    &   0    &   0  &  0  &  0  &  0  &  0 \\
			0  & 0    &   0    &   0  &  0  &  0  &  0  &   0 \\
		\end{bmatrix}^T,
	\end{matrix}
	\]
	
	\[
	{\bm \Theta_{0}} =64\begin{bmatrix}
		0  & 1    &   0     &  1   &  0  &  1  &  0  &   1  \\
		0  & 0    &   0     &  0  &  0  &  0  &  0  &  0  \\
		0  & 0    &   0    &   0  &  0  &  0  &  0  &  0 \\
		0  & 0    &   0    &   0  &  0  &  0  &  0  &   0 \\
	\end{bmatrix}^T,
\]
\[
	\begin{matrix}
			{\bm \Theta_{1}}=64\begin{bmatrix}
			1  & 0    &   1     &  0  &   1  &  0  &  1  &   0  \\
			0  & 0    &   0     &  0  &  0  &  0  &  0  &  0  \\
			0  & 0    &   0    &   0  &  0  &  0  &  0  &  0 \\
			0  & 0    &   0    &   0  &  0  &  0  &  0  &   0 \\
		\end{bmatrix}^T.
	\end{matrix}
	\]

	By the inverse FZT, we can obtain the cross-correlation ${\bm \theta_{1,2}}$ and the auto-correlation ${\bm \theta_{1}}$ and ${\bm \theta_{2}}$ from ${\bm \Theta_{1,2}}$ ${\bm \Theta_{1}}$ and ${\bm \Theta_{2}}$, respectively:
	\[
	{\bm \theta_{0,1}}={\bm 0},
	\]
	\[
	\begin{split}
	{\bm \theta_{0}}=32&(1,0,0,0,0,0,0,0,0,0,0,0,0,0,0,0,\\
	&-1,0,0,0,0,0,0,0,0,0,0,0,0,0,0,0),
	\end{split}
	\]
	\[
	\begin{split}
	{\bm \theta_{1}}=32&(1,0,0,0,0,0,0,0,0,0,0,0,0,0,0,0,\\
	&1,0,0,0,0,0,0,0,0,0,0,0,0,0,0,0).
	\end{split}
	\]
	It is clear that the sequence set $\{\ubu_a\mid0\leq a \leq 1 \}$ is an optimal $(32,2,15)$ IF-ZCZ sequence set.
\end{example}

\begin{remark}
	Note that the IF-ZCZ set proposed \cite[Example 3]{Brozik_A} cannot reach the Tang-Fan-Matsufuji bound with the sequence length $N=32$ and set size $K=2$, while the sequence set in Example 2 from our construction is optimal with the same length and size parameters.
\end{remark}

\section{Computational Complexity}\label{5}
Both constructions $\bm{S}_1$ and $\bm{S}_2$ with sparse and highly structured Zak have certain advantages in terms of the implementations complexity of the corresponding banks of matched filters. In order to demonstrate these advantages, the total number of complex multiplications and complex additions will be used as the implementation complexity evaluation. Typically the DFT/IDFT is performed by some Fast Fourier Transform (FFT) algorithm. The usual approximation of the number of complex multiplications and complex additions in an $N$-points FFT/IFFT algorithm is $(N/2){\rm log_2}N$ and $N{\rm log_2}N$, respectively.

\begin{algorithm}[!htbp]
	\caption{A Zak-Domain Implementation}
	\label{alg:C}
	\begin{algorithmic}
		\REQUIRE ~~\\
		
		Compute $L\times M$ FZTs of $K$ reference sequences, yielding
		\[
		U_a(j,k)=\sum_{r=0}^{L-1}u_a(k+rM)\omega_L^{rj},
		\]
		 where $0\leq j\leq L-1,\ 0\leq k \leq M-1,0\leq a \leq K-1$.
		
		\ENSURE ~~\\
		
		Step 1 : Compute $L\times M$ FZT of the vector of input  samples $\ybu$, yielding
		\[
		Y(j,k)=\sum_{r=0}^{L-1}y(k+rM)\omega_L^{rj},
		\]
		 where $0\leq j\leq L-1,\ 0\leq k \leq M-1$.
		
		Step 2 : Perform ZS crosscorrelations $\theta_{y,u}$ of $y$ and  the $K$ reference sequences, the ZS crosscorrelation of $\ybu$ and $\sbu_u$ is defined by
		\[
		\Theta_{\ybu,\ubu_a}(j,k)=\sum_{l=0}^{M-1}Y(j,l)U_a^*(j,l-k).
		\]
		where $0\leq j\leq L-1,\ 0\leq k \leq M-1,0\leq a \leq K-1$,
		
		Step3 : Compute the inverse FZT of the results of previous step, the
		crosscorrelation of $\ybu$ and $\sbu_a$
		\[
		\theta_{\ybu,\ubu_a}(k+rM)=\frac{1}{N}\sum_{j=0}^{L-1}\Theta_{\ybu,\ubu_a}(j,k)\omega_L^{-rj}
		\]
		where $0\leq r \leq L-1,\ 0\leq k \leq M-1,0\leq a \leq K-1$.
	\end{algorithmic}
\end{algorithm}
For a set of $K$ matched filter banks corresponding to the IF-ZCZ sequence set, step 1 only needs to be performed once, and the $L$-point DFT is calculated $M$ times to obtain the Zak domain array $\bm{Y}$, which requires $(N/2){\rm log}_2L$ complex multiplications and $N{\rm log}_2L$ complex additions. Generally speaking, in step 2, each matched filter needs to calculate L times $M$-point convolution, $K$ matched filters need $KLM$ complex multiplications, and $KL(M-1)$ complex multiplications addition. In step 3, each matched filter needs to calculate $M$ times $L$-L point IDFT, which requires $K(N/2){\rm log}_2L$ complex multiplications and $KN{\rm log}_2L $ times complex addition. 

If the Zak domain matched filter algorithm is implemented based on the construction $\bm{S_1}$, then step 2 only needs to perform the $M$ point convolution operation on the $K$ matched filters, because the construction $\bm{S_1}$ The IF-ZCZ sequence has only $M$ non-zero FZT coefficients and is located in a fixed row in Zak space with a unique maximum magnitude value. In step 3, you only need to find the position of the maximum amplitude value in step 2 to determine. In summary, the computational complexity is
\[
\begin{split}
	O_1^1(N,K)
	&=M\cdot\frac{3}{2}K{\rm log}_2K+K\cdot M(2M-1)\\
	&=\frac{3}{2}N{\rm log}_2K+\frac{2N^2}{K}-N.
\end{split}
\]

If the Zak domain matched filter algorithm is implemented based on the construction $\bm{S_2}$, then step 2 only requires bit shifting of $K$ matched filters. Each sequence has only $M$ nonzero coefficients of 1 FZTs in different columns and equally spaced rows in Zak space. In step 3, it is only necessary to perform IFZT on the Zak spectrum obtained in step 2 with only $M$ rows not zero. In summary, the computational complexity is
\[
\begin{split}
	O_1^2(N,K)
	&=M\cdot\frac{3}{2}K{\rm log}_2K+K\cdot \frac{3}{2}M\log_2M\\
	&=\frac{3}{2}N{\rm log}_2N.
\end{split}
\]

In general, a direct time-domain implementation of a matched filter involves $N$ point convolution of $N$ consecutive received signal samples $\ybu$ with a reference sequence $\ubu_a$ of length $N$, requiring $N^2$ complex multiplications and N(N-1) complex additions. Therefore, a set of K matched filters associates the same input sample vector with K reference sequences, and the total computational complexity of a set of K matched filters based on the time domain is
\[
O_2(N,K)=KN(2N-1).
\]

To illustrate how big the savings in receiver complexity can be if the construction $\bm{S_1}$ and $\bm{S_2}$ are deployed. In the table below, we give the values of $O_1^1(N,K)$, $O_1^2(N,K)$ and
$O_2(N,K)$. 

\begin{table}[H]
	\centering
	\caption{三种算法的计算复杂度比较}
	\begin{tabular}{|c|c|c|c|c|}
		\hline
		\textbf{Period $N$}
		&128
		&128
		&128
		&128\\
		\hline
		\textbf{Size $K$}
		&32
		&16
		&8
		&4\\
		\hline
		\textbf{$O_1^1(N,K)$}
		&926
		&2668
		&4544
		&8848\\
		\hline
		\textbf{$O_1^2(N,K)$}
		&1344
		&1344
		&1344
		&1344\\
		\hline
		\textbf{$O_2(N,K)$}
		&1044480
		&522240
		&261120
		&130560\\
		\hline
	\end{tabular}
\end{table}

\section{Conclusion}\label{6}
It has been shown that all known constructions of optimal IF-ZCZ sequence sets are included in the general construction \cite{Popovic_IF-ZAZ}. We have proven that the general construction is equivalent to a simple construction of optimal IF-ZCZ sequence sets. It is implied that the alphabet size for the special case of the simple construction of the optimal IF-ZCZ sequence set can be a factor of the period. The computational complexity of implementing matched filter banks based on the simple construction of the optimal IF-ZCZ sequence set and its special case can be reduced.
Whether there is an optimal IF-ZCZ sequence set that is not equivalent to the general construction \cite{Popovic_IF-ZAZ} is an open question.


\begin{thebibliography}{1}
	
	
	\bibitem{Brozik_A} A. K. Brodzik, ``On Certain Sets of Polyphase Sequences With Sparse and Highly Structured Zak and Fourier Transforms,'' {\em IEEE Trans. Inf. Theory}., vol. 59, no. 10, pp. 6907-6916, Oct. 2013.
	
	\bibitem{Brozik_B} A. Brodzik, ``Polyphase Golay sequences with semi-polyphase Fourier transform and all-zero crosscorrelation: Construction B, `` in {\em Excursions in Harmonic Analysis}, vol. 3. Basel,
	Switzerland: Birkhäuser, 2015, pp. 211–229. [Online]. Available: http://www.springer.com/gp/book/9783319132297
	
	\bibitem{Brozik_Bat}  A. K. Brodzik and R. Tolimieri, ``Bat chirps with good properties: Zak
	space construction of perfect polyphase sequences,'' {\em IEEE Trans. Inf. Theory}., vol. 55, no. 4, pp. 1804-1814, Apr. 2009.
	
	\bibitem{Brozik_Characterization} A. K. Brodzik, ``Characterization of Zak space support of the discrete chirp,'' {\em IEEE Trans. Inf. Theory}., vol. 53, no. 6, pp. 2190-2203, Jun. 2007.
	
	\bibitem{Chung} H. Chung and P. V. Kumar, “A new general construction for generalized bent functions,” {\em IEEE Trans. Inform. Theory}, vol. IT-35, pp. 206-209, 1989.
	
	\bibitem{Coker} J. D. Coker and A. H. Tewfik, ``Simplified ranging systems using discrete wavelet decomposition,'' {\em IEEE Trans. Signal Process.}, vol. 58, no. 2, pp. 575-582, Feb. 2010.
	
	
	
	\bibitem{Fan1} P. Z. Fan, ``Spreading sequence design and theoretical limits for
	quasisynchronous CDMA systems,'' {\em EURASIP J. Wireless Commun.
		Netw.}, vol. 2004, no. 1, pp. 19-31, 2004.
	
	\bibitem{4G} ${\rm 3^{rd}}$ Generation Partnership Project, Technical Specification Group Radio Access Network, {\em Evolved Universal Terrestrial Radio Access (E-UTRA), Physical Channels and Modulation}, document 3GPP TS 36.211, v14.3.0, Jun. 2017, sec. 5.7.2, p. 69.
	
	
	
	\bibitem{Islam} K. M. Z. Islam, T. Y. Al-Naffouri, and N. Al-Dhahir, ``On optimum
	pilot design for comb-type OFDM transmission over doubly-selective
	channels,'' {\em IEEE Trans. Commun.}, vol. 59, no. 4, pp. 930-935, Apr.
	2011.
	
	\bibitem{Kumer} P. V. Kumar, R. A. Scholtz and L. R. Welch, “Generalized bent functions and their properties,” {\em J. Combin. Theory}, Series A, vol. 40, pp. 90-107, 1985.
	
	
	\bibitem{Jiang} X. Y. Jiang, ``Code hopping communications for anti-interception with real-valued QZCZ sequences,'' {\em IEEE Trans. Commun.}, vol. 59, no. 3, pp. 680-685, Mar. 2011.
	
	
	\bibitem{Milewski} A. Milewski, “Periodic sequences with optimal properties for channel estimation and fast start-up equalization,” {\em IEM J. Res. Develop}., vol. 27, pp. 426-431, 1983.
	
	\bibitem{Mow} W. H. Mow, ``A new unified construction of perfect root-of-unity sequences,'' in {\em Proc. IEEE 4th Int. Symp. Spread Spectr. Techn. Appl. (ISSSTA)}, Sep. 1996, pp. 955-959
	
	\bibitem{Matsufuji} S. Matsufuji, N. Kuroyanagi, N. Suehiro, and P. Z. Fan, ``Two types of polyphase sequence sets for approximately synchronized CDMA systems, `` {\em IEICE TRANS. Fundam. Electron., Commun. Comput. Sci.}, vol. E86-A, no. 1, pp. 229–234, 2003.
	
	
	
	\bibitem{Popovic_IF-ZAZ} B. M. Popovi$\rm{\acute{c}}$, ``Optimum Sets of Interference-Free Sequences WIth Zero Autocorrelation Zones", {\em IEEE Trans. Inf. Theory}., vol. 64, no. 4, pp. 1406-1409, Apr 2018.
	
	\bibitem{Popovic chirp like} B. M. PopoviC, “Generalized chirp-like polyphase sequences with optimum correlation properties,” {\em IEEE Trans. Inform. Theory}, vol. IT-38, pp. 1406-1409, 1992.
	
	\bibitem{Patole}  S. M. Patole, M. Torlak, D. Wang, and M. Ali, ``Automotive radars: A review of signal processing techniques,'' {\em IEEE Signal Process. Mag.}, vol. 34, no. 2, pp. 22-35, Mar. 2017.
	
	
	\bibitem{Rajamohan} N. Rajamohan and A. P. Kannu, ``Downlink synchronization techniques for heterogeneous cellular networks,'' {\em IEEE Trans. Commun}., vol. 63, no. 11, pp. 4448-4460, Nov. 2015.
	
	
	\bibitem{TangBound} X. H. Tang, P. Z. Fan, and S. Matsufuji, ``Lower bounds on the maximum correlation of sequence set with low or zero correlation zone,'' {\em Electron. Lett.}, vol. 36, pp. 551-552, Mar. 2000.
	
	
	
	\bibitem{Welch} L. R.Welch, ``Lower bounds on the maximum cross correlation of signals,'' {\em IEEE Trans. Inf. Theory.}, vol. 20, no. 3, pp. 397-399, May 1974.
	
	
	\bibitem{Zak} J. Zak, ``Finite translations in solid state physics,'' {\em Phys. Rev. Lett}., vol. 19, pp. 1385-1397, 1967.
	
	
	
\end{thebibliography}
\end{document}